\documentclass[letterpaper, 10 pt, conference]{ieeeconf}  

\IEEEoverridecommandlockouts                              

\usepackage{amssymb}
\usepackage{amsfonts}
\usepackage{amsbsy}
\usepackage{amsmath}

\usepackage{algorithm2e}

\newtheorem{definition}{Definition}[section]

\newtheorem{Lemma}{Lemma}[section]
\newtheorem{theorem}{Theorem}[section]

\newtheorem{remark}{Remark}[section]
\newtheorem{example}{Example}[section]

\newcommand{\bmat}{\left[ \begin{array}}
\newcommand{\emat}{\end{array} \right]}
\newcommand{\F}{{\mathbb F}}

\newcommand{\N}{\mathbb{N}}
\newcommand{\C}{{\cal C}}

\newcommand{\im}{\textnormal{Im }}

\title{\LARGE \bf
Erasure codes with simplex locality
}

\author{Margreta Kuijper$^{1}$ and Diego Napp$^{2}$
\thanks{*D. Napp's research  has been supported by the Spanish grant DPI2012-31509 and by a ``Juan de la Cierva'' grant (JCI-2010-06268).}
\thanks{$^{1}$Margreta Kuijper is with Faculty of Electrical and Electronic Engineering of the University of Melbourne, Australia
        {\tt\small mkuijper@unimelb.edu.au}}%
\thanks{$^{2}$Diego Napp is with the Department of Mathematics, University Jaume I,
        Castellon, Spain
        {\tt\small napp@uji.es}}%
}

\begin{document}

\maketitle
\thispagestyle{empty}
\pagestyle{empty}

\begin{abstract}

We focus on erasure codes for distributed storage. The distributed storage setting imposes locality requirements because of "easy repair" demands on the decoder. We first establish the characterization of various locality properties in terms of the generator matrix of the code. These lead to bounds on locality and notions of optimality. We then examine the locality properties of a family of nonbinary codes with simplex structure. We investigate their optimality and design several "easy repair" decoding methods. In particular, we show that any correctable erasure pattern can be solved by easy repair.

\end{abstract}

\section{Introduction}

Several classical coding techniques (such as Reed-Solomon erasure codes) are
extensively used for data storage, most successfully applied to storage in RAID systems and magnetic recording
(see \cite{Blaum}). However, due to the fast-growing demand for large-scale data storage, it would be impossible or extremely expensive to build single pieces of
hardware with enough storage capabilities to store the enormous volume of data that is being
generated. Hence, new classes of storage technology have emerged using the idea of
distributing data across multiple nodes which are interconnected over a network,
as we witness in some peer-to-peer (P2P) storage systems \cite{Wuala} and data centers \cite{Azure} that comprise
the backbone infrastructure of cloud computing. We call such systems Networked Distributed
Storage Systems (NDSS).

A fundamental issue that arises in this context is the so-called \emph{Repair Problem}: how to maintain the
encoded data when failures (node erasures) occur. When a storage node fails, information that was stored in the node is no longer accessible. As a remedy a node is then added to the
system to replace the failed node. The added node downloads data from a set of appropriate and accessible nodes to recover the information stored in the failed node. This is called {\em node repair}. To
assess the performance of this repair process, there are several metrics that can be considered:
{\em storage cost}, measured as the amount of data stored in the node, {\em repair bandwidth}, measured as the total
number of bits communicated in the network for each repair and {\em locality}, measured as the number of
nodes needed for each repair. For instance, $(n,k)$ maximum distance separable (MDS) codes are
optimal in terms of storage cost since any $k$ nodes contain the minimum amount
of information required to recover the original data. However, to repair one single node it is
necessary to retrieve information from all $k$ nodes. More specifically, repair is achieved by re-encoding the information from these $k$ nodes and storing part of the re-encoded data in the new node.
This results in a poor performance with respect to repair bandwidth as well as locality.
%

Currently the most well-understood metrics are the repair bandwidth metric and the storage cost metric, see for example~\cite{Dimakis1,Gaston}. Using network coding techniques, several code constructions have been presented that show optimality with respect to repair bandwidth and storage cost, see \cite{Dimakis2} and
references therein. In contrast, locality is an important metric that has received less attention in the literature. This metric was studied independently by several authors, see \cite{Oggier1}, \cite{Gopalan} and \cite{Dimakis3} among others,
and it is considered to be one of the main repair performance bottlenecks in many NDSS, e.g., in
cloud storage applications.

\begin{definition}
An $(n,k)$ code has {\em locality} $r$ if every codeword symbol in a codeword is a linear combination of at most $r$ other symbols in the codeword.
\end{definition}

Thus, when a code of locality $r$ is used then one needs to contact at most $r$ nodes to
repair one node. In the recent paper~\cite{Gopalan} (see also~\cite{TamoBarg}) it was shown that there exists a natural trade-off among
redundancy, locality and code minimum distance:
\begin{theorem}
Let $C$ be an $(n,k)$ linear code with minimum distance $d$ and locality $r$. Then
$$
n-k +1 - d \geq \left\lfloor \frac{k-1}{r} \right\rfloor .
$$
\end{theorem}
\vspace{.2cm}
\begin{proof}(from~\cite{TamoBarg}; see also~\cite{Gopalan})
Let $G$ be the generator matrix of $C$. Choose any $\left\lfloor \frac{k-1}{r} \right\rfloor $ nodes of the code---call these the "leaders". Each leader can be written as a linear combination of at most $r$ other nodes---call this set the "set of friends of the leader". Now define $N$ as the set of nodes which is the union of all sets of friends of the leaders but without the leaders themselves. Then clearly $N$ has less than $k$ elements so that the set of columns in $G$ that corresponds to $N$ spans a space of rank $< k$. Since $G$ has full rank it is possible to enlarge $N$ to a set $N'$ of $\geq k-1$ columns such that the rank of its corresponding columns equals exactly $k-1$. Note that because the code has locality $r$, this enlargement operation can be done without involving any of the leaders. Now define $U$ as the union of $N'$ and the set of leaders. Then $U$ has at least $k-1+\left\lfloor \frac{k-1}{r} \right\rfloor$ nodes but still, because the code has locality $r$, the corresponding columns in $G$ span a space of dimension $<k$. By definition of the minimum distance, all ($k\times \cdot$)-submatrices of $G$ that have rank $<k$ must have $\leq n-d$ columns. It therefore follows that $k-1+\left\lfloor \frac{k-1}{r} \right\rfloor \leq n-d$ which proves the theorem.
\end{proof}
 From the above bound it is seen that MDS codes do not perform well with respect to locality. Indeed, since $d=n-k+1$ they have only trivial locality $r=k$. In contrast, non-MDS codes such as Pyramid codes and Hierarchical codes have been shown to be optimal with respect to the above bound. For these codes the gap $n-k +1 - d$ is nonzero due to the fact that they are not MDS. In a sense these codes optimally "use" this gap for locality purposes. More generally, a new class of codes, called locally repairable codes (LRC)~\cite{Dimakis3,Dimakis4,Oggier1,KamPra}, addresses the repair problem focusing on minimizing the number of nodes contacted during the repair process.

However, the issue of locality for {\em multiple} node erasures is less well researched.
More specifically, the main problem with the existing locally repairable codes is that although they minimize the number of contacted nodes for the case that only one node has failed, they suffer from the drawback that it is not known how many nodes are needed when several failures occur. This can be due to, for instance, the fact that in these constructions only a single subset of nodes can repair a particular piece of redundant data and therefore if a node from this repair subset is also not available, data cannot be repaired "locally", increasing the cost of the repair. Hence, it is desirable to obtain codes providing multiple repair alternatives. Some interesting preliminary results on this problem have been recently presented in \cite{PraKam}, seeking to extend the ideas in \cite{Gopalan}, and in \cite{Pamies} using partial geometry. It is also worth mentioning the results in \cite{Oggier1,Oggier2} where some code schemes, akin to the one presented here, are introduced.

\begin{definition}
Let $\delta$ be a positive integer and $C$ be an $(n,k)$ code with erasure correcting capability $\leq \delta$. Then $C$ has {\em $\delta$-locality} $r$ if, for any erasure pattern with $\leq \delta$ erased symbols, every erased codeword symbol is a linear combination of at most $r$ live symbols in the codeword.
\end{definition}

In this paper, we present a coding scheme based on \emph{simplex codes} especially suitable when locality is relevant and multiple failures may occur. We show how it is always possible to repair each node by simply adding two live nodes, even in the presence of multiple erasures.

\section{Preliminaries}

Let $q=2^m$ and let $\F_q$ be a finite field with $q$ elements. A $q$-ary linear $(n,k)$-code $\C$ of length $n$ and rank $k$ is a $k$-dimensional linear subspace of $\F_q^n$. Full-rank matrices $G\in \mathbb{F}_q^{k \times n}$ and $H \in \mathbb{F}_q^{n \times (n-k)}$ with the property that
\begin{eqnarray*}
{\mathcal C}  &=& \im_{\F_q}G = \left\{ c= u G \in \mathbb F_q^{n}:\, u \in \mathbb F_q^{k}\right\} \\
&=& \ker_{\F_q}H = \left\{ c  \in \mathbb F_q^{n}:\, Hc^\top=0 \right\},
\end{eqnarray*}
are called \textit{generator matrix} and \textit{parity-check matrix}, respectively.

\vspace{.4cm}

\begin{definition}
Let $k$ be a positive integer, $n= 2^k -1$ and let $G$ be a $k  \times  n $ matrix whose columns are the distinct non-zero vectors of $\F^k_2$. Let $C$ be the binary code over $\F_2$ that has $G$ as its generator matrix. Then $C$ is called a binary \textit{simplex} $(n,k)$-code.
\end{definition}

\vspace{.4cm}

Binary simplex codes are classical codes with minimum distance $d = 2^{k-1}$ and they are dual to the binary Hamming codes.


\vspace{.4cm}

\begin{example}\label{ex:01}
Let $k=3$ and $\C\subset \F_2^7$ be a simplex $(7,3)$-code. Then, its generator matrix $G$ and parity-check matrix $H$ are given by;
$$
G= \left(
     \begin{array}{ccccccc}
       1 & 0 & 0 & 1 & 1 & 0 & 1 \\
       0 & 1 & 0 & 1 & 0 & 1 & 1 \\
       0 & 0 & 1 & 0 & 1 & 1 & 1 \\
     \end{array}
   \right)
$$
and
$$ \ H = \left(
                    \begin{array}{ccccccc}
                       1 & 1 & 0&1 & 0 & 0 & 0\\
                       1 & 0 & 1&0 & 1 & 0 & 0\\
                       0 & 1 & 1&0 & 0 & 1 & 0\\
                       1 & 1 & 1&0 & 0 & 0 & 1\\
                    \end{array}
                  \right),
$$
respectively.
\end{example}

\vspace{.4cm}

When coding is used in distributed storage systems, a data object or file $u=(u_1,u_2,\dots,u_k)\in \F_q^k$ of $k$ symbols is redundantly stored across $n$ different nodes in $c=(c_1,c_2,\dots,c_n)= u G\in \F_q^n$, where $G$ is the generator matrix of an $(n,k)$-code $\C$.

Let $S= \{c_1,\dots,c_n  \}$ be the set of nodes, $S^e\subset S$ the set of erased notes, $S^\ell\subset S$ the set of live nodes and $S_i^e= \{ i \ | \ c_i \in S^e \}$ and $S_i^\ell= \{ i \ | \ c_i \in S^\ell \}$  the indices of the erased and live nodes, respectively. A node $c_i$ is said to be \emph{related} to the pair $(c_j,c_k)$ if $c_i=c_j + c_k$, $i,j,k\in \{1,\dots,n \}$. Two pairs of nodes are said to be \emph{disjoint} if they do not share a common node. If $c_i\in S^e$ is related to the pair $(c_j,c_k)$ where $c_j,c_k \in S^\ell$, then it is said that $c_i$ \emph{allows for easy repair}.

In terms of computational complexity, this implies that the cost of a node reconstruction is that of a simple addition of two nodes.

Note that if $c_i$ is related to the pair $(c_j,c_k)$, then $c_j$ is related to $(c_i,c_k)$ and $c_k$ is related to $(c_j,c_i)$.

The following lemma is useful for the sequel of the paper.
\begin{Lemma}\label{le:aux}
Let $C$ be a linear $(n,k)$ code and let $S_i^e$ denote the set of indices of erased nodes. Denote $|S_i^e| = n-s$. Then the following statements are equivalent:
\begin{enumerate}
  \item The erasure pattern corresponding to $S_i^e$ is correctable;
  \item The $k \times s$ matrix $\hat G$ formed by deleting the $i$-th columns of $G$ where $i\in S_i^e$, is right invertible;
  \item The $(n-k) \times (n-s)$ matrix $\hat H$ formed by the $i$-th columns of $H$ where $i\in S_i^e$, is left invertible.
\end{enumerate}
\end{Lemma}

\vspace{.2cm}

\begin{proof}
Denote the set of indices of the live (=non-erased) nodes by $S_i^\ell$. Clearly 1) holds if and only if there do not exist two different codewords that coincide in positions corresponding to $S_i^\ell$. Since the code is linear this is equivalent to the non-existence of a nonzero codeword whose symbols at positions in $S_i^\ell$ are zero. The latter is clearly equivalent to the linear independence of the columns of $\hat H$. Next, we prove the equivalence of 1) and 2). Write $S_i^\ell = \{ j_1 , \ldots , j_s\}$. Consider the system of equations
 \[
 \bmat{ccc} c_{j_1} & \cdots & c_{j_s}\emat= u \hat G .
   \]
 The solvability of this system is equivalent to the recovery of $u$ and the repair of all erasures. The equivalence of 1) and 2) now follows from the fact that this system is solvable for any erasure pattern that corresponds to $S_i^e$ if and only if $\hat G$ is right invertible.

\end{proof}

\section{Simplex locality}

In this section we propose a nonbinary simplex code, defined as follows:

Let $G\in \F_2^{k \times n}$ and $H\in \F_2^{n \times (n-k)}$ be the generator matrix and parity-check matrix of a binary simplex $(n,k)$-code over $\F_2$. Via this generator matrix $G$ we encode the data to be stored $u\in \F_q^k$ to
\begin{equation}\label{eq:01}
    c=(c_1,c_2,\dots,c_n)= u G\in \F_q^n,
\end{equation}
with $q=2^m$ for some $m\in \N$. The resulting code is an $(n,k)$ code over $F_q$ that we call a {\em simplex code over $F_q$}.

It is easy to see that these codes inherit their distance property from the binary simplex codes, namely $d = 2^{k-1}$. The codes also possess several good locality properties, starting with the next lemma which is based on a wellknown property of the binary simplex code.

\begin{Lemma}\label{lem:01}
Let $C$ be an $(n,k)$ simplex code over $F_q$. Denote its set of nodes by $S= \{c_1,\dots,c_n  \}$. Then, each node $c_i\in S$, $i\in \{ 1,\dots, n \}$ is related to $\frac{n-1}{2}$ different pairs.
\end{Lemma}
\begin{proof}
Choose any node, say $\hat c\in S$, and let $\hat S :=\{ \hat c\}$. Take $c_{1_j}\in S \setminus\hat S$ and set $c_{1_k} = c_{1_j} + \hat c$. Due to the fact that any sum of two columns of $G$ is another column of $G$ we have that for any $c_j,c_k\in S$, $c_j+c_k \in S$. Hence, $c_{1_k}\in S$ and let $\hat S := \hat S \bigcup \{ c_{1_j}, c_{1_k}\}$. Again take any node $c_{2_j}\in S \setminus\hat S$ and set $c_{2_k}= c_{2_j} + \hat c \in S$. Note that $c_{2_k}\notin \hat S$ and therefore the pairs $(c_{1_j},c_{1_k})$ and $(c_{2_j},c_{2_k})$ are disjoint. Let $\hat S := \hat S \bigcup \{ c_{2_j}, c_{2_k}\}$ and the cardinality of the set $\hat S$ is increased by two in each step. Repeating this process $\frac{n-1}{2}$ times, we obtain $\frac{n-1}{2}$ disjoint pairs related to $\hat c$. Since the choice of $\hat c$ is arbitrary, this concludes the proof.
\end{proof}

It follows from the above lemma that a simplex code over $F_q$ has locality $2$. Therefore any single erasure pattern allows for easy repair. The next theorem, reminiscent of Corollary 3 in~\cite{Oggier1}, shows that this is also true for multiple erasure patterns that are within the code's erasure correcting capability.

\begin{theorem}
Let $C$ be an $(n,k)$ simplex code over $F_q$. Denote the set of erased nodes by $S^e$. If $|S^e| \leq \frac{n-1}{2}$, then \emph{all} the nodes in $S^e$ allow for easy repair. Thus, the $\frac{n-1}{2}$-locality of $C$ equals $2$.
\end{theorem}
\begin{proof}
If $|S^e| \leq \frac{n-1}{2}$ then we have $n-|S^e| > \frac{n-1}{2}$ live nodes. By Lemma \ref{lem:01} any erased node is related to $\frac{n-1}{2}$ disjoint pairs which implies that at least one of these pairs is comprised of two live nodes, which implies easy repair.
\end{proof}

Note that the above theorem deals with erasure patterns whose number of erasures are within the erasure correcting capability of the code. We now turn our attention to the larger class of erasure patterns that are correctable and possibly have $> \frac{n-1}{2}$ erasures.

\vspace{.4cm}

\begin{Lemma}\label{exists}
Let $C$ be an $(n,k)$ simplex code over $F_q$. Denote the set of erased nodes by $S^e$. Then if $S^e$ corresponds to a correctable erasure pattern then there exists an erased node that allows for easy repair.
\end{Lemma}

\vspace{.2cm}

\begin{proof}
Denote again the set of live nodes by $S^\ell$; denote the set of its indices by $S_i^\ell=\{ j_1 , \ldots , j_s\}$.
Consider the system of equations
 \[
 \bmat{ccc} c_{j_1} & \cdots & c_{j_s}\emat= u \hat G ,
  \]
where $\hat G$ is the $k \times s$ matrix that remains after deleting from $G$ the $i$-th columns at erased positions. Since $S^e$ corresponds to a correctable erasure pattern, this system of equations is solvable over $\F_q$. Thus it follows from Lemma~\ref{le:aux} that $\hat G$ has rank $k$. As a result,
for all $\hat c\in S^e$ there exist an integer $g$ and $a_j\in S_i^\ell$ for $j=1, \dots, g$ such that
$$
\hat c= c_{a_1} +\dots + c_{a_g}.
$$
If $c_{a_1} + c_{a_2} \in S^e$, then we have found one erased node that is easily reparable. If not, {\em i.e.} if $c_{a_1} + c_{a_2} \in S^\ell$, then denote $c_{b_1}=c_{a_1} + c_{a_2}$ and therefore $\hat c= c_{b_1} + c_{a_3} +\dots + c_{a_g}$. Again, if $c_{b_1} + c_{a_3} \in S^e$, then we have found one erased node that is easily reparable. If not, {\em i.e.}, if $c_{b_1} + c_{a_3} \in S^\ell$, then denote $c_{b_2}=c_{b_1} + c_{a_3}$ and therefore $\hat c= c_{b_2} + c_{a_4} +\dots + c_{a_g}$. This process must end yielding either that $c_{b_j} + c_{a_{j+2}} \in S^e$ with $c_{b_j}, c_{a_{j+2}} \in S^\ell$ for some $j\in\{ 1,2, \dots, g-3\}$ or $\hat c= c_{b_{g-2}} + c_{a_g}$. In both cases we obtain an easy repair and the proof is completed.
\end{proof}

\vspace{.2cm}

In the following algorithm we present an ``easy repair of one node" algorithm for an encoded file $c=(c_1,\dots,c_n)$ with node failures. The generator matrix used for the codification is $G=\left[
                                     \begin{array}{ccc}
                                       G_1 & \cdots & G_n \\
                                     \end{array}
                                   \right]
 $ where its columns $G_i$ are elements of $\F_{2^k}$. 


\begin{algorithm}\label{alg}
 \KwData{$(c_1,\dots,c_n)$ and $G=\left[
                                     \begin{array}{ccc}
                                       G_1 & \cdots & G_n \\
                                     \end{array}
                                   \right]
 $. \\
 }
 \KwResult{$(c_{j_i}, c_{j_t},c_{j_i} + c_{j_t} , \mbox{Correctable})$. \\
  }

 $( j_1, \dots, j_{s} )=findlivenodes(c)$\;
 $s=length(findlivenodes(c))$\;
 $G^\ell= \{ G_{j_1}, \dots, G_{j_{s}} \}$\;  $i=1,\ t=2$ \;

 \While{ $[G_{j_i} + G_{j_t} \in G^\ell]$ AND $[i< s-1]$}{
  \eIf{$t=s$}{$ t=i+2, \ i=i+1$}{$t=t+1$}}
 \eIf{$i=s$}{$(c_{j_{i-1}}=0, c_{j_i}=0, c_{j_{i-1}} + c_{j_i}=0$, $\mbox{Correctable=FALSE})$}{
 $(c_{j_i} , c_{j_t}, c_{j_i} + c_{j_t},\mbox{Correctable=TRUE})$}

\end{algorithm}

Above, $findlivenodes(c)$ is a function that returns a vector $( j_1, \dots, j_{s} )$ of live nodes indices. The algorithm returns two live nodes $c_{j_i}$ and $c_{j_t}$ that repair an erased node $c_{j_i} + c_{j_t}$ and the Boolean variable  ``Correctable" that takes the values TRUE or FALSE; in case the erasure pattern is uncorrectable, the algorithm returns zero values and the statement ``FALSE"

\vspace{.2cm}

\begin{theorem}\label{th:01}
Let $C$ be an $(n,k)$ simplex code over $F_q$ with $S^e$ denoting the set of erased nodes. 
If $S^e$ corresponds to a correctable erasure pattern, then repeated application of Algorithm \ref{alg} recovers all erasures by easy repairs.
\end{theorem}
\begin{proof}
Because of Lemma \ref{exists}, an easy repair situation exists. The algorithm is clearly defined in such a way that it will find this easy repair situation and carry out the repair. Once repaired, the erasure pattern is of course still correctable and we repeat over and over until all erasures are recovered.
\end{proof}

\vspace{.2cm}

\begin{example}
Consider the matrices $G$ and $H$ as defined in Example \ref{ex:01}. Suppose that we have a file $u\in \F_q^3$ to be stored in $7$ nodes, \emph{i.e.}, $u G= c=(c_1,\dots,c_7)\in \F_q^7$, and that erasures occur in nodes $c_1,c_2,c_4$ and $c_6$, \emph{i.e.}, $S_i^e= \{ 1,2,4,6\}$, $S_i^\ell= \{ 3,5,7\}$. Thus the pattern is correctable despite the fact that the number of erasures is outside of the erasure correcting capability. It now follows from Lemma~\ref{exists} that there exists an erased node that allows for easy repair. Indeed, $c_1=c_3 + c_5$ and in fact there exist several erased nodes that allow for easy repair, namely also $c_2= c_5 + c_7$ and $c_4= c_3 + c_7$. By Lemma~\ref{exists}, we can repair all nodes by easy repairs. Indeed, $c_2$ and $c_4$ already allow for easy repair and, once $c_2$ is repaired, we repair $c_6$ from $c_6=c_2+ c_3$.
\end{example}
\vspace{.2cm}
\begin{remark}
Note that in the previous example the node $c_6$ cannot be the first node to be easily repaired and we need to repair a different node first. However, when the number of erasures does not exceed the erasure correcting capability of the code, then \emph{any} erased node can be chosen to start the repair, thus allowing for parallelization of easy repairs.
\end{remark}
\vspace{.2cm}
\section{Conclusions}

In this paper, we have presented the family of non-binary simplex codes as a family which is highly suitable for efficient erasure coding in multiple-erasure settings within Networked Distributed Storage Systems. These non-binary codes are constructed using the generator matrix of binary simplex codes and hence inherit excellent locality properties even in the presence of large erasure patterns.
They allow for easy encoding, requiring only addition operations in $F_q$. We have shown that if the erasure pattern is correctable at all, then it is possible to repair each of the failed nodes at the cost of a simple addition of two live nodes. An algorithm for such an easy repair is provided. Further analysis to evaluate these simple codes in practical situations against existing code families is part of our ongoing and future work.

\end{document}